\documentclass{LMCS}

\def\dOi{9(3:10)2013}
\lmcsheading%
{\dOi}
{1--22}
{}
{}
{Oct.~\phantom05, 2010}
{Sep.~\phantom04, 2013}
{}

\usepackage{enumerate,hyperref}

\usepackage[utf8]{inputenc}
\usepackage{afterpage,mathpartir,amsmath,amsfonts,amsthm,amssymb,%
stmaryrd,cmll,enumerate, microtype,graphicx,empheq} %
\usepackage[mathcal]{euscript}
\usepackage[nohug,midshaft,tight]{diagrams}
\usepackage{version,MnSymbol}

\newarrow{to}----{->}
\newarrow{into}{boldhook}---{->}
\newarrow{implies}===={=>}
\newarrow{modto}--+-{->}

\usepackage{tikz,tikz-cats}
\usepackage{hyperref}

\DeclareMathOperator{\dom}{dom}

\newarrow{eqto}=====

\newcommand{\bigcat}[1]{\ensuremath{\mathsf{#1}}}

\newcommand{\commentthis}[1]{}

\newcommand{\C}{\mathcal{C}}
\newcommand{\D}{\mathcal{D}}
\newcommand{\F}{\mathcal{F}}

\newcommand{\Fi}{\mathcal{F}_1}

\renewcommand{\L}{\mathcal{L}}
\renewcommand{\H}{\mathcal{H}}

\newcommand{\Hi}{\mathcal{H}_1}

\newcommand{\J}{\mathcal{J}}
\newcommand{\K}{\mathcal{K}}
\newcommand{\Ki}{\mathcal{K}_1}

\newcommand{\M}{\mathcal{M}}

\newcommand{\U}{\mathcal{U}}

\newcommand{\Ui}{\mathcal{U}_1}

\newcommand{\V}{\mathcal{V}}

\newcommand{\Vi}{\mathcal{V}_1}

\newcommand{\W}{\mathcal{W}}
\newcommand{\Wi}{\W_1}

\newcommand{\X}{\mathcal{X}}

\newcommand{\R}{\mathcal{R}}

\renewcommand{\S}{\mathcal{S}}

\newcommand{\StdSigi}{\bigcat{Sig}_{1}} 
\newcommand{\StdSig}{\bigcat{Sig}}

\newcommand{\Set}{\bigcat{Set}}

\newcommand{\cccat}{\bigcat{CCCat}}

\newcommand{\tcccat}{\bigcat{2CCCat}}

\newcommand{\tccc}{cartesian closed 2-category}
\newcommand{\tcccs}{cartesian closed 2-categories}

\newcommand{\abs}{\ell}
\newcommand{\app}{a}

\newcommand{\un}{I}

\newcommand{\para}{_{||}}

\newcommand{\transl}[1]{\llbracket #1 \rrbracket}
\newcommand{\rond}{\circ}

\newcommand{\id}{\mathit{id}}
\newcommand{\iso}{\cong}
\newcommand{\impll}{\multimap}
\newcommand{\tens}{\otimes}

\newcommand{\card}[1]{|#1|}

\newcommand{\ens}[1]{\{ #1 \}}

\newcommand{\alt}{\mathrel{|}}

\newcommand{\lam}{\ensuremath{\lambda}}

\newcommand{\subs}[1]{[#1]}

\newcommand{\ev}{\mathit{ev}}

\newcommand{\Gam}{\Gamma}
\newcommand{\Del}{\Delta}

\newcommand{\rcolon}{\mathrel{:}}
\newcommand{\red}{\mathrel{\to}}

\newcommand{\req}{\mathrel{\equiv}}

\newcommand{\vcomp}[1]{\mathbin{;_{#1}}}
\newcommand{\proj}{\pi}
\newcommand{\projj}{\pi'}

\newcommand{\Lo}{\ensuremath{\L_0}}
\newcommand{\So}{\ensuremath{\S_0}}

\newcommand{\Li}{\ensuremath{\L_1}}

\newcommand{\ruleset}[2]{
{\ } \\
\mbox{#1} \\ \hrulefill
  \begin{mathpar}
    #2
  \end{mathpar}
}
\newcommand{\bureaucratic}[1]{}

\newcommand{\abar}{\overline{a}}
\newcommand{\LAlg}{\L\mbox{-}\mathrm{Alg}}

\newcommand{\LiAlg}{\L_1\mbox{-}\mathrm{Alg}}

\numberwithin{equation}{section} \numberwithin{paragraph}{section}

\newcommand{\constapp}[2]{#1 \llparenthesis #2 \rrparenthesis}
\newcommand{\ruleapp}[2]{#1 \llangle #2 \rrangle}
\newcommand{\pairing}[1]{\langle #1 \rangle}
\newcommand{\unitadj}{\eta^\L}

\def\framed{%
\setbox0=\vbox\bgroup%
\advance\hsize by -2\fboxsep\advance\hsize by -2\fboxrule%
\linewidth=\hsize%
}
\def\endframed{%
\egroup\noindent\framebox[\textwidth]{\box0}\vspace*{1mm}}

\begin{document}

\title[Cartesian closed 2-categories and permutation
equivalence]{Cartesian closed 2-categories and permutation equivalence
  in higher-order rewriting}

\author[T.~Hirschowitz]{Tom Hirschowitz} 
\address{CNRS, Université de Savoie} 
\email{tom.hirschowitz@univ-savoie.fr}
\thanks{Partially
    funded by the French ANR \emph{projets blancs} PiCoq ANR-10-BLAN-0305 and
    R\'ecr\'e ANR-11-BS02-0010.}



%
\ACMCCS{[{\bf Theory of computation}]: Semantics and reasoning---Program semantics---Categorical semantics / Denotational semantics / Operational semantics}


\begin{abstract}
  We propose a semantics for permutation equivalence in higher-order
  rewriting.  This semantics takes place in cartesian closed
  2-categories, and is proved sound and complete.
\end{abstract}

\maketitle

\section{Introduction} Cartesian closed categories provide semantics
for equational theories with variable
binding~\cite{Lambek:categorical,CROct}.  On the other hand,
2-categories with finite products provide semantics for term
rewriting~\cite{DBLP:conf/rta/CorradiniGM95}.  The present paper shows
that \emph{cartesian closed 2-categories} provide semantics for term
rewriting with variable binding, as embodied by Brugginks's
generalisation~\cite{bruggink:phd2008} of \emph{permutation
  equivalence}~\cite[Chapter 8]{terese} to \emph{higher-order
  rewriting}~\cite{KlopCRS,Wolfram,DBLP:conf/lics/Nipkow91,DBLP:conf/hoa/OostromR93}.

We first define \emph{cartesian closed 2-signatures}, which generalise
higher-order rewrite systems, and organise them into a category
$\StdSig$.  We then construct an adjunction
\begin{equation}
  \adj{\StdSig}{\tcccat,}{\H}{\W}
  \label{adj:main}
\end{equation}
where $\tcccat$ is the category of small cartesian closed
2-categories.  From a given cartesian closed 2-signature $S$, the
functor $\H$ constructs a \tccc{}, whose 2-cells are Bruggink's proof
terms modulo permutation equivalence, which we prove is the free
\tccc{} generated by $S$.

We review a number of examples and non-examples, and sketch an
extension to deal with the latter.

\subsection*{Related work}
Our cartesian closed 2-signatures may be seen as a 2-dimensional refinement of
\emph{cartesian closed
  sketches}~\cite{DBLP:journals/tcs/Wells90,Despeyroux,KPT:sketches}.
Bruggink's calculus of permutation equivalence is close in spirit to
Hilken's 2-categorical semantics of the simply-typed
$\lam$-calculus~\cite{DBLP:journals/tcs/Hilken96}, but technically
different and generalised to arbitrary higher-order rewrite systems.
Capriotti~\cite{Capriotti} proposes a semantics of so-called
\emph{flat} permutation equivalence in sesquicategories.  More related
work is discussed in Section~\ref{subsec:binding}.

\section{Cartesian closed signatures and categories}
\label{sec:ccc}
We start by recalling the well-known
adjunction~\cite{Lambek:categorical,CROct} between what we here call
\emph{(cartesian closed) 1-signatures} and cartesian closed
categories.

For any set $X$,
define \emph{types} over $X$ by the grammar:
$$
\begin{array}[c]{rcll}
  A, B, \ldots \in \Lo (X) & ::= & x \alt 1 \alt A \times B \alt B^A,
\end{array}
$$
with $x \in X$.

\begin{prop}
  $\Lo$ defines a monad on $\Set$.
\end{prop}

Let the set of \emph{sequents} over a set $X$ be $\So (X) = \Lo
(X)^* \times \Lo (X)$, i.e., sequents are pairs of a list of types and
a type. The assignment $X \mapsto \So (X)$ extends to an endofunctor
on $\Set$.

\begin{defi}
  A \emph{1-signature} consists of a set $X_0$ of \emph{sorts}, and an
  \linebreak $\So (X_0)$-indexed set $X_1$ of \emph{operations}, or
  equivalently a map $X_1 \to \So (X_0)$.
\end{defi}
A \emph{morphism of 1-signatures} $(X_0, X_1) \to (Y_0, Y_1)$ is a
pair $(f_0, f_1)$ where $f_i \colon X_i \to Y_i$ such that
  \begin{center}
    \diag (1,3) {
      X_1 \& Y_1 \\
      \So(X_0) \&       \So(Y_0) 
    }{ \sq{f_1}{}{}{\So (f_0)} %
    }
  \end{center}
  commutes.  Morphisms compose in the obvious way, and we have:
  \begin{prop}
    Composition of morphisms is associative and unital, and hence
    1-signatures and their morphisms form a category $\StdSigi$.
  \end{prop}

  There is a well-known adjunction 
  \begin{center}
    \adj{\StdSigi}{\cccat}{\Hi}{\Wi}
  \end{center}
  between 1-signatures and the category $\cccat$ of small cartesian
  closed categories (with chosen structure) and (strict) cartesian
  closed functors, i.e., functors $F \colon \C \to \D$ preserving
  binary products, projections, and the terminal object on the nose, and such that,
  for all objects $A,B \in \C$,
  currying 
  $$F (B^A) \times F(A) = F(B^A \times A) \xrightarrow{F(\ev_{A,B})} F(B)$$
  yields an identity.

  The functor $\Wi$ maps any cartesian closed category $\C$ to the
  signature with sorts $\C_0$, its set of objects, and with operations
  $A_1, \ldots, A_n \to A$ the set $\C (\transl{A_1 \times \ldots
    \times A_n}, \transl{A})$, where $\transl{-}$ denotes the function
  $\Lo (\C_0) \to \C_0$ defined by induction:

  \begin{equation}
  \begin{array}{rcll}
    \transl{c} & = & c & \mbox{\ \ \ $c \in \C_0$} \\
    \transl{1} & = & 1 \\
    \transl{A \times B} & = & \transl{A} \times \transl{B} \\
    \transl{B^A} & = & \transl{B} ^ {\transl{A}}.    
  \end{array}\label{eq:transl}
\end{equation}

Conversely, given a 1-signature $X$, consider the simply-typed
$\lam$-calculus with base types in $X_0$ and constants in $X_1$. I.e.,
for any $c \in X_1(G,A)$ and terms $\Gam \vdash M_i \colon A_i$, for
all $1 \leq i \leq n$, where $G = (A_1, \ldots, A_n)$, we have a term
$\Gam \vdash \constapp{c}{M_1, \ldots, M_n} \colon A$, representing
the application of the constant $c$ to $M_1, \ldots, M_n$. We use
special parentheses to avoid ambiguity with term application.  Terms
modulo $\beta\eta$ form a cartesian closed category $\Hi(X)$ with
objects all types over $X_0$ and morphisms $A \to B$ all terms of type
$B$ with one free variable of type $A$.

  A less often formulated observation, which is useful to us, is that
  the adjunction $ \Hi \dashv \Wi$ decomposes into two adjunctions
  \begin{center}
    \adjs{\StdSigi}{\LiAlg}{\Ki}{\Ui}{\cccat,}{\Fi}{\Vi}
  \end{center}
  as follows.

  Consider first the endofunctor $\Li$ on $\StdSigi$ defined on objects
  by mapping any 1-signature $X$ to the 1-signature with
  \begin{iteMize}{$\bullet$}
  \item 
    as sorts the set $X_0$, and
  \item as operations $\Gam \vdash A$ the $\lam$-terms $\Gam \vdash M
    \colon A$, with base types in $X_0$ and constants in $X_1$, as
    sketched above, modulo $\beta\eta$.
  \end{iteMize}
  On morphisms of 1-signatures $X \xrightarrow{f} Y$, let $\Li (f)$
  substitute constants $c \in X_1$ with $f_1 (c)$.  We obtain
  \begin{prop}
    $\Li$ is a monad on $\StdSigi$, with unit and multiplication, say $\eta$ and $\mu$.
  \end{prop}

  Let now $\LiAlg$ be the category of algebras for the monad $\Li$
  and $\Ki$ be the `free algebra' functor $X \mapsto (\Li (X),
  \mu_X)$.  
  
  The functor $\Vi$ maps any cartesian closed category $\C$ to the
  $\Li$-algebra with base 1-signature $(\C_0, \C_1)$, defined as
  follows. First, $\C_0$ is the set of objects of $\C$. It has a
  canonical $\Lo$-algebra structure, say $h_0 \colon \Lo (\C_0) \to
  \C_0$, obtained by interpreting type constructors in $\C$ as
  in~\eqref{eq:transl}. Extending this to contexts $G$ by $h_0 (G) =
  \prod_{i} h_0 (G_i)$, let the operations in $\C_1 (G, A)$ be the
  morphisms in $\C(h_0(G), h_0(A))$. Beware: the domain and codomain of
  such an operation are really $G$ and $A$, not $h_0 (G)$ and $h_0
  (A)$. Similarly, interpreting the $\lam$-calculus in $\C$, the
  1-signature $(\C_0, \C_1)$ has a canonical $\Li$-algebra structure,
  say $h_1 \colon \Li (\C_0, \C_1) \to (\C_0, \C_1)$:
  $$
  \begin{array}{rcll}
    h_1 (G \vdash x_i \colon G_i) & = & \proj_i \\
    h_1 (G \vdash () \colon 1) & = & ! \\
    h_1 (G \vdash \constapp{c}{M_1, \ldots, M_n} \colon A) & = &
    c \rond \pairing{h_1 (M_1), \ldots, h_1 (M_n)} \\
    h_1 (G \vdash \lam x \colon A. M \colon B^A) & = & %
    \varphi (h_1 (G, x \colon A \vdash M \colon B)) \\
    h_1 (G \vdash M N \colon B) & = & \ev \rond \pairing{h_1(M), h_1(N)} \\
    h_1 (G \vdash (M, N) \colon A \times B) & = & \pairing{h_1(M), h_1(N)} \\
    h_1 (G \vdash \proj M \colon A) & = & \proj \rond h_1(M) \\
    h_1 (G \vdash \projj M \colon A) & = & \projj \rond h_1(M),
  \end{array}
  $$
  where $!$ is the unique morphism $h_0 (G) \to 1$, $\varphi$ is the 
  bijection $\C (h_0 (G, A), h_0(B)) \iso \C(h_0 (G), h_0(B^A))$, and 
  $\ev$ is the structure morphism $h_0(B^A \times A) \to h_0(B)$.

  $\Li$-algebras are much like cartesian closed categories whose
  objects are freely generated by their set of sorts. A perhaps useful
  analogy here is with multicategories $\M$, seen as being close to
  monoidal categories whose objects are freely generated by those of
  $\M$ by tensor and unit. Here, the functor $\Fi$ sends any $\Li$-algebra $(X, h)$
  to the cartesian closed category with
  \begin{iteMize}{$\bullet$}
  \item objects the types over $X_0$, i.e., $\Lo (X_0)$,
  \item morphisms $A \to B$ the set of operations in $X_1 (A, B)$.
  \end{iteMize}
  This canonically forms a cartesian closed category, with structure
  induced by the $\Li$-algebra structure. We define it in more detail
  in dimension 2 in Section~\ref{subsec:left}.
  
\section{Cartesian closed 2-signatures}

Given a 1-signature $X$, let $X\para{}$ denote the set of pairs of
parallel operations, i.e., pairs of operations $M, N$ over the same
sequent. Otherwise said, $X\para{}$ is the pullback
\begin{center}
  \Diag{%
    \pbk{m-2-1}{m-1-1}{m-1-2} %
    }{%
    X\para{} \& X_1 \\
    X_1 \& \So(X_0).
  }{ %
    (m-1-1) edge (m-1-2) %
    (m-1-1) edge (m-2-1) %
    (m-2-1) edge (m-2-2) %
    (m-1-2) edge (m-2-2) %
    }
\end{center}

Any morphism $f \colon X \to Y$ of 1-signatures yields a map $f\para{} \colon X\para{} \to Y\para{}$, via the dashed arrow (obtained by universal property of pullback) in
\begin{center}
  \diagramme[diagorigins={1}{1.4}]{}{%
    \pbk{m-3-1}{m-1-1}{m-1-3} %
    \pbk{m-4-2}{m-2-2}{m-2-4} %
    \path[->] 
    (m-1-1) edge[fore,dashed,->] (m-2-2) %
    ; %
    }{%
      X\para{} \& \& X_1 \&\\
      \& Y\para{} \& \& Y_1 \\
      X_1 \& \& \So(X_0) \\
      \& Y_1 \& \& \So(Y_0).
      }{%
        (m-1-1) edge (m-1-3) %
        (m-1-1) edge (m-3-1) %
        (m-1-3) edge (m-3-3) %
        (m-3-1) edge (m-3-3) %
        (m-2-2) edge[fore,->] (m-2-4) %
        (m-2-2) edge[fore,->] (m-4-2) %
        (m-2-4) edge (m-4-4) %
        (m-4-2) edge (m-4-4) %
        (m-1-3) edge[labelar={f_1}] (m-2-4) %
        (m-3-3) edge[labelbl={\So (f_0)}] (m-4-4) %
        (m-3-1) edge[labelbl={f_1}] (m-4-2) %
        }
\end{center}

\begin{defi}
  A \emph{2-signature} consists of a 1-signature $X$, plus a set $X_2$
  of \emph{reduction rules} with a map $X_2 \to \Li (X)\para{}.$
\end{defi}

A \emph{morphism of  2-signatures} $(X, X_2) \to (Y, Y_2)$ is a pair $(f, f_2)$ where
$f \colon X \to Y$ is a morphism of 1-signatures and $f_2 \colon X_2 \to Y_2$ makes 
the diagram
\begin{center}
  \diag{%
    X_2 \& Y_2 \\
    \Li(X)\para{} \& \Li(Y)\para{}
    }{%
      (m-1-1) edge[labelu={f_2}] (m-1-2) %
      (m-1-1) edge (m-2-1) %
      (m-1-2) edge (m-2-2) %
      (m-2-1) edge[labeld={\Li(f_1)\para{}}] (m-2-2)
    }
\end{center}
commute. We obtain:
\begin{prop}
  Composition of morphisms is associative and unital, and hence
  2-signatures and their morphisms form a category $\StdSig$.
\end{prop}

\section{Examples}

\subsection{Higher-order rewrite systems}
The prime example of a 2-signature is that for the pure
$\lam$-calculus: it has a sort $t$ and operations
  \begin{mathpar}
    \app \colon (t, t \vdash t) \and \abs \colon (t^t \vdash t),
  \end{mathpar}
  with a reduction rule $\beta$ over the pair 
  $x \colon t^t, y \colon t \vdash \constapp{\app}{\constapp{\abs}{x}, y}, x (y) \colon t$
  in $\Li (\ens{t}, \ens{\abs, \app})\para{}$.
  Categorically, this will yield a 2-cell
  \begin{center}
    \Diag(.5,1){%
      \path %
      (u) edge[twocell,labelr={\beta}] (d) %
      ; %
    }{%
      \& |(u)| t \times t \& \\
      t^t \times t \& \& t.  }{%
      (m-2-1) edge[labelal={\abs \times t}] (u) %
      (u) edge[labelar={\app}] (m-2-3) %
      (m-2-1) edge[bend right=15,twod={\ev}] (m-2-3) %
    }
  \end{center}

  This is an example of a \emph{higher-order rewrite system} in the
  sense of Nipkow~\cite{DBLP:conf/lics/Nipkow91}.  Nipkow's definition
  is formally different, but his higher-order rewrite systems are in
  bijection with 2-signatures $h \colon X_2 \to \Li (X)\para{}$ such
  that for all rules $r \in X_2$, letting $(\Gam \vdash M, N \colon A)
  = h (r)$:
    \begin{iteMize}{$\bullet$}
    \item $M$ is not a variable,
    \item $A$ is a sort,
    \item each variable occurring in $\Gam$ occurs free in $M$.
    \end{iteMize}
    These restrictions help dealing with decidability problems on
    higher-order rewrite systems, whose extension to our setting we
    leave open.

    Let us now anticipate over Adjunction~\eqref{adj:main} and our
    main results below and state our soundness and completeness
    theorem. Given a higher-order rewrite system $X$, i.e., a
    2-signature satisfying the above conditions, let $\R (X)$ be the
    following locally-preordered 2-category. It has:
    \begin{iteMize}{$\bullet$}
    \item objects are types in $\Lo (X_0)$;
    \item morphisms $A \to B$ are $\lam$-terms in $\Li (X) (A \vdash
      B)$, modulo $\beta\eta$;
    \item given two parallel morphisms $M$ and $N$, there is one
      2-cell $M \to N$ exactly when there is a sequence of reductions
      $M \red^* N$ in the usual sense~\cite{DBLP:conf/lics/Nipkow91}.
    \end{iteMize}

    \begin{prop}
      $\R (X)$ is 2-cartesian closed.
    \end{prop}

    $\R (X)$ and $\H (X)$ have the same objects and morphisms. But
    because our inference rules for forming reductions are the same as
    deduction rules for proving the existence of a reduction in the
    usual sense, we may map any reduction $P \colon M \red N$ to the
    unique reduction $M \red N$ in $\R (X)$. 
    Conversely, any standard reduction step has a proof, which provides 
    a reduction $P$. We have proved:
    \begin{thm}[Soundness and completeness]
      There exists an identity-on-objects, identity-on-morphisms,
      locally full cartesian closed 2-functor $\H (X) \xrightarrow{!}
      \R (X)$.
    \end{thm}

\subsection{Theories with binding}\label{subsec:binding}
Understanding reduction rules as equations, it is easy to define the
free cartesian closed category generated by a 2-signature. This yields
the adjunction 
\begin{equation}
  \adj{\StdSig}{\cccat}{\H'}{\W'}
  \label{eq:adj2}
\end{equation}
recalled above.

This adjunction provides a categorical semantics for theories with
binding, which is more general than other approaches by Fiore and
Hur~\cite{DBLP:conf/lics/Fiore08}, Hirschowitz and
Maggesi~\cite{hirscho:lam}, and Zsidó~\cite{Zsido}.

If I understand correctly, the motivation for Fiore and Hur's subtle
approach is the will to explain the $\lam$-calculus by strictly less
than itself. The present framework does not obey this specification,
and instead tends to view the $\lam$-calculus as a universal
(parameterised) theory with binding. 

We end this section by giving a formal construction of
Adjunction~\eqref{eq:adj2}. Cartesian closed categories form a full,
reflective subcategory of $\tcccat$, via the functor $\J \colon
\tcccat \to \cccat$ mapping any cartesian closed 2-category $\C$ to the
cartesian closed category with:
\begin{iteMize}{$\bullet$}
\item objects those of $\C$,
\item morphisms those of $\C$, modulo the congruence generated by $f
  \sim g$ iff there exists a 2-cell $f \to g$.
\end{iteMize}
Here, $\J (\C)$ is thought of as the free locally discrete cartesian
closed 2-category.  Adjunction~\eqref{eq:adj2} is obtained by
composing the adjunctions
\begin{center}
  \adjs{\StdSig}{\tcccat}{\H}{\W}{\cccat.}{\J}{}
\end{center}

\subsection{Non-examples}
Non-examples are given by calculi whose reduction semantics is defined
on terms modulo a so-called \emph{structural congruence}, e.g.,
CCS~\cite{Milner80}, or the
$\pi$-calculus~\cite{Engberg86acalculus,Milner:pi}.
  
  For example, consider the CCS term $(a \alt 0) \alt \abar$.  In CCS, it is
  \emph{structurally} equivalent to $(a \alt \abar) \alt 0$, which then
  reduces to $0 \alt 0$.

  In order to account for this, we would have to consider a
  2-signature with reduction rules for structural congruence, here
  $(M_1 \alt M_2) \alt M_3 \red M_1 \alt (M_2 \alt M_3)$ for
  associativity, and $M \alt N \red N \alt M$ for commutativity. But
  then, these reductions count as proper reductions, which departs
  from the desired computational behaviour. For example, the term $a
  \alt a$ has an infinite reduction sequence, using commutativity.

  Anticipating the development in the next sections, a potential
  solution is to extend 2-signatures to \emph{2-theories}. For any
  2-signature $X$, let $X\para{}$ denote the set of pairs of reduction
  rules $r, s$ with a common type $G \vdash M \red N \colon A$.  A
  2-theory is a 2-signature $X$, together with a set of equations
  between parallel reductions, i.e., a subset $X_3$ of $\L(X)\para{}$
  (where $\L$ is defined in Section~\ref{sec:L}).

  Another possibility would be to define 2-theories to constist of a
  pair of sets over $\Li (X)\para{}$, one for structural equations,
  and the other for proper reduction rules.

  The main adjunction announced above~\eqref{adj:main} extends to an
  adjunction between 2-theories and cartesian closed 2-categories.
  Using equations, we may specify that any reduction $M \red M$ using
  only structural rules be the identity on $M$, and consider the
  computational behaviour of a 2-category to consist of its
  non-invertible 2-cells, as proposed by
  Hilken~\cite{DBLP:journals/tcs/Hilken96}.  A question is whether for
  a given calculus this can be done with finitely many equations.

\section{A 2-lambda-calculus}\label{sec:L}
We now begin the construction of Adjunction~\eqref{adj:main}.
We start in this section by defining a monad $\L$ on $\StdSig$, which 
we will use to factor Adjunction~\eqref{adj:main} as
\begin{center}
  \adjs{\StdSig}{\LAlg}{\K}{\U}{\tcccat{},}{\F}{\V}
\end{center}
where:
\begin{iteMize}{$\bullet$}
\item $\LAlg$ is the category of $\L$-algebras,
\item $\K \colon \StdSig \to \LAlg$ maps any $X$ to the free $\L$-algebra
  $(\L^2X \xrightarrow{\mu} \L X)$,
\item $\U (\L X \xrightarrow{h} X) = X$,
\item $\tcccat{}$ is the category of cartesian closed 2-categories,
  which we define in Section~\ref{sec:tcccat}.
\end{iteMize}

The left-hand adjunction holds by $\L$ being a monad, thus we
concentrate in Section~\ref{sec:main} on establishing the right-hand
one.

But for now, let us define the monad $\L$. 
\subsection{Syntax}
Given a 2-signature $X =
((X_0, X_1), a \colon X_2 \to \Li (X)\para{})$ (actually
$\Li (X)$ is $\Li (X_0, X_1)$), we construct a new 2-signature $\L
(X)$, whose reduction rules represent reduction sequences in the
``higher-order rewrite system'' defined by $X$, modulo permutation
equivalence.  The 2-signature $\L (X)$ has the same base 1-signature
$(X_0, X_1)$, and as reduction rules the terms of a 2$\lam$-calculus
(in the sense of Hilken~\cite{DBLP:journals/tcs/Hilken96}) modulo
permutation equivalence, which we now define.

\begin{figure}[t]
  \centering
  \begin{framed}
    \begin{mathpar}
      \inferrule{
        \ldots \\ \Gam \vdash P_i \rcolon M_i \red N_i \colon G_i \\ \ldots \\
      }{ \Gam \vdash \ruleapp{r}{P_1, \ldots, P_n} \rcolon M \subs{M_1, \ldots,
          M_n} \red N \subs{N_1, \ldots, N_n} \colon A }~ (r \in X (G
      \vdash M, N \colon A)) \and
      \inferrule{ %
        \Gam \vdash P \rcolon M_1 \red M_2 \colon A \\
        \Gam \vdash Q \rcolon M_2 \red M_3 \colon A }{%
        \Gam \vdash P \vcomp{M_2} Q \rcolon M_1 \red M_3 \colon A }%
      \and %
      \\
      \inferrule{}{ \Gam, x \colon A, \Del \vdash x \rcolon x \red x
        \colon A } \and %
      \inferrule{} {\Gam \vdash () \rcolon () \red () \colon 1} \and %
      \inferrule{
        \Gam \vdash P_1 \rcolon M_1 \red N_1 \colon G_1 \\ \ldots \\
        \Gam \vdash P_n \rcolon M_n \red N_n \colon G_n \\
      }{ \Gam \vdash \constapp{c}{P_1, \ldots, P_n} \rcolon \constapp{c}{M_1, \ldots, M_n}
        \red \constapp{c}{N_1, \ldots, N_n} \colon A}~(c \in X_1 (G \vdash A)) %
      \and
      \inferrule{ \Gam, x \colon A \vdash P \rcolon M \red N \colon B
      }{ \Gam \vdash \lam x \colon A . P \rcolon \lam x \colon A. M
        \red \lam x \colon A. N \colon B^A }
      \and %
      \inferrule{ %
        \Gam \vdash P \rcolon M \red M' \colon B^A \\ %
        \Gam \vdash Q \rcolon N \red N' \colon A %
      }{ %
        \Gam \vdash P Q \rcolon M N \red M' N' \colon B } %
      \and %
      \inferrule{%
        \Gam \vdash P \rcolon M \red M' \colon A \\ %
        \Gam \vdash Q \rcolon N \red N' \colon B %
      }{ %
        \Gam \vdash (P, Q) \rcolon (M, N) \red (M', N') \colon A
        \times B } %
      \and \\%
      \inferrule{ %
        \Gam \vdash P \rcolon M \red N \colon A \times B %
      }{ %
        \Gam \vdash \proj_{A,B} P \rcolon \proj_{A,B} M \red
        \proj_{A,B} N \colon A %
      } %
      \and %
      \inferrule{ %
        \Gam \vdash P \rcolon M \red N \colon A \times B %
      }{ %
        \Gam \vdash \projj_{A,B} P \rcolon \projj_{A,B} M \red
        \projj_{A,B} N \colon B %
      } %
    \end{mathpar}

  \end{framed}
  \caption{Reductions}
  \label{fig:reductions}
\end{figure}

First, terms, called \emph{reductions}, are defined by induction in
Figure~\ref{fig:reductions}. The typing judgement has the shape $\Gam
\vdash P \rcolon M \red N \colon A$, where $A$ is a type in $\Lo
(X_0)$, $\Gam$ is a list of pairs of a variable and a type, with no
variable appearing more than once, $M$ and $N$ are terms of type $\Gam
\vdash A$ modulo $\beta\eta$, and $P$ is a reduction. In the sequel,
we often forget the variables in such pairs $(\Gam \vdash A)$, and
identify them with sequents in $\So (X_0)$.

\begin{rem}
  For any $\Gam \vdash M \colon A$, we have a reduction $\Gam \vdash M
  \rcolon M \red M \colon A$.
\end{rem}

When clear from context, we abbreviate substitutions $\subs{M_1 / x_1,
  \ldots, M_n / x_n}$ of terms by $\subs{M_1, \ldots, M_n}$. For a
context $G$, $G_i$ denotes its $i$th type. Also, for $(M, N) \in \Li
(X)\para{}$, we let $X (M, N)$ be the set of all reduction rules $r
\in X_2$ such that $a (r) = (M, N)$. We write $X (\Gam \vdash M, N
\colon A)$ to indicate the common type of $M$ and $N$. Similarly, $X
(G \vdash A)$ denotes the set of operations in $X_1$ above $G \vdash
A$.

\subsection{Substitution}
Next, we define substitution, which has ``type''
\begin{equation}
  \inferrule{
    \Gam \vdash Q \rcolon N \red N' \colon \Del \\
    \Del \vdash P \rcolon M \red M' \colon A
    }{
      \Gam \vdash P[Q] \rcolon M[N] \red M'[N'] \colon A,
    }\label{eq:typesubst}
\end{equation}
i.e., given a reduction $P$ and a tuple of reductions $Q$, it produces
a reduction of the indicated type, which we denote by $P[Q]$.  Here, we
denote by $\Gam \vdash Q \rcolon N \red N' \colon \Del$ a tuple of
reductions $\Gam \vdash Q_i \rcolon N_i \red N'_i \colon \Del_i$, for
$1 \leq i \leq \card{\Del}$.

First, observe that we have a form of weakening: for any reduction
$\Gam \vdash P \rcolon M \red N \colon A$ and $x \notin \Gam$, we also
have $\Gam, x \colon B \vdash P \rcolon M \red N \colon A$.
We use this implicitly in the following.

The definition of substitution is a bit tricky:
\begin{iteMize}{$\bullet$}
\item first we define \emph{left whiskering}, which has ``type''
\begin{mathpar}
  \inferrule{
    \Gam \vdash Q \rcolon N \red N' \colon \Del \\
    \Del \vdash M \colon A
    }{
      \Gam \vdash M[Q] \rcolon M[N] \red M'[N'] \colon A;
    }
\end{mathpar}
\item then we define \emph{right whiskering}, which has ``type''
\begin{mathpar}
  \inferrule{
    \Gam \vdash  N \colon \Del \\
    \Del \vdash P \rcolon M \red M' \colon A
    }{
      \Gam \vdash P[N] \rcolon M[N] \red M'[N] \colon A,
    }
\end{mathpar}
(where $N$ denotes a tuple);
\item then we define substitution by $P[Q] = (P[N] \vcomp{M'[N]} M'[Q]).$
\end{iteMize}
There is of course another legitimate definition, namely
$M[Q] \vcomp{M[N']} P[N'].$
The two will be equated by permutation equivalence in the next
section.


Left whiskering is defined inductively, with $\Del = (x_1 \colon A_1,
\ldots, x_n \colon A_n)$ and $Q = (Q_1, \ldots, Q_n)$, by:
\[\begin{array}[t]{rcll}
  ()[Q] & = & () \\
  x_i[Q] & = & Q_i \\
  \constapp{c}{M_1, \ldots, M_p}[Q] & = & \constapp{c}{M_1[Q], \ldots, M_p[Q]} \\
  (\lam x \colon B. M)[Q] & = & \lam x \colon B. (M[Q, x]) \mbox{\ \ (for $x \notin \dom (\Del)$)} \\
  (M N)[Q] & = & (M[Q] N[Q]) \\
    (M, N)[Q] & = & (M[Q], N[Q]) \\
    (\pi_{A,B} M)[Q] & = & \pi_{A,B} (M[Q]) \\
    (\pi'_{A,B} M)[Q] & = & \pi'_{A,B} (M[Q])
\end{array}\]
Right whiskering is defined inductively, with $\Del = (x_1 \colon A_1, \ldots, x_n \colon A_n)$ and
$N = (N_1, \ldots, N_n)$, by:
\[\begin{array}[t]{rcll}
  (\ruleapp{r}{P_1, \ldots, P_p})[N] & = & \ruleapp{r}{P_1[N], \ldots, P_p[N]} \\
  (P_1 \vcomp{M''} P_2)[N] & = & (P_1[N] \vcomp{M''[N]} P_2[N]) \\
  ()[N] & = & () \\
  x_i[N] & = & N_i \\
  \constapp{c}{P_1, \ldots, P_p}[N] & = & \constapp{c}{P_1[N], \ldots, P_p[N]} \\
  (\lam x \colon B. P')[N] & = & \lam x \colon B. (P'[N, x]) \mbox{\ \ (for $x \notin \dom (\Del)$)} \\
  (P_1 P_2)[N] & = & (P_1[N] P_2[N]) \\
    (P_1, P_2)[N] & = & (P_1[N], P_2[N]) \\
    (\pi_{A,B} P')[N] & = & \pi_{A,B} (P'[N]) \\
    (\pi'_{A,B} P')[N] & = & \pi'_{A,B} (P'[N]).
\end{array}\]


\begin{defi}
  Let $P[Q] = (P[N] \vcomp{M'[N]} M'[Q])$.
\end{defi}
\begin{prop}
  Given reductions $P$ and $Q$ as above, $P[Q]$ is a well-typed
  reduction $\Gam \vdash P[Q] \rcolon M[N] \red M'[N'] \colon A$.
\end{prop}

\subsection{Permutation equivalence}
We now define \emph{permutation equivalence} on reductions, by the
equations in Figures~\ref{fig:req}, \ref{fig:req1} and~\ref{fig:req2}, in
Appendix~\ref{sec:req}. The \emph{congruence} rules in
Figure~\ref{fig:req} are bureaucratic: they just say that permutation
equivalence is a congruence. The \emph{category} rules make reductions
of a given type $\Gam \vdash A$ into a category.  In
Figure~\ref{fig:req1}, the \emph{beta and eta} rules mirror the
term-level beta and eta rules. Finally, the \emph{lifting} rules lift
composition of reductions towards toplevel.

So, $\L (X)$ has sorts $X_0$, operations $X_1$, and as reduction rules
in $\L (X) (G \vdash M, N \colon A)$ all reductions $G \vdash P
\rcolon M \red N \colon A$, modulo the equations.

This easily extends to:
\begin{prop}
  $\L$ is a functor $\StdSig \to \StdSig$.
\end{prop}

Now, consider $\L\L (X)$. We define a mapping $\mu_X \colon \L\L (X) \to
\L (X)$, by induction on reductions.
The typing rule for reduction
rules in $\L\L (X)$ specialises to:
\begin{mathpar}
    \inferrule{  (R \in \L(X) (G \vdash M, N \colon A)) \\
      \Gam \vdash P_1 \rcolon M_1 \red N_1 \colon G_1 \\ \ldots \\ 
      \Gam \vdash  P_n \rcolon M_n \red N_n \colon G_n  \\
    }{
      \Gam \vdash \ruleapp{R}{P_1, \ldots, P_n} \rcolon M \subs{M_1, \ldots, M_n}
      \red N \subs{N_1, \ldots, N_n} \colon A
    }\ \cdot
\end{mathpar}
We set $\mu(\ruleapp{R}{P_1, \ldots, P_n}) = R[\mu(P_1), \ldots, \mu(P_n)]$.
The other cases just propagate the substitution:
$$
\begin{array}{rcll}
  P \vcomp{} Q & \mapsto & \mu (P) \vcomp{} \mu (Q) \\
  x & \mapsto & x \\
  ()  & \mapsto & () \\
  \constapp{c}{P_1, \ldots, P_n} & \mapsto & \constapp{c}{\mu(P_1), \ldots, \mu(P_n)} \\
  \lam x \colon A. P  & \mapsto & \lam x \colon A. \mu (P) \\
  PQ  &\mapsto & \mu (P) \mu (Q) \\ 
  (P, Q)  &\mapsto & (\mu (P), \mu (Q)) \\ 
  \proj P  & \mapsto & \proj (\mu (P)) \\
  \projj P  & \mapsto & \projj (\mu (P)).
\end{array}
$$

\begin{lem}
  This defines a natural transformation $\mu \colon \L^2 \to \L$, which makes the 
  diagram
  \begin{center}
    \diag{%
      \L^3 \& \L^2 \\
      \L^2 \& \L %
    }{%
      (m-1-1) edge[labelu={\L\mu}] (m-1-2) %
      (m-1-1) edge[labell={\mu\L}] (m-2-1) %
      (m-1-2) edge[labelr={\mu}] (m-2-2) %
      (m-2-1) edge[labeld={\mu}] (m-2-2) %
    }
  \end{center}
  commute.
\end{lem}
Similarly, there is a natural transformation $\eta \colon \id \to \L$,
mapping each $r \in X (G \vdash M, N \colon A)$ to the reduction $G
\vdash \ruleapp{r}{x_1, \ldots, x_n} \rcolon M \red N \colon A$, and we have:
\begin{lem}
  The diagram
  \begin{center}
    \diag{%
      \L \& \L^2 \& \L \\
       \& \L %
    }{%
      (m-1-1) edge[labelu={\eta\L}] (m-1-2) %
      (m-1-1) edge[identity] (m-2-2) %
      (m-1-3) edge[identity] (m-2-2) %
      (m-1-3) edge[labelu={\L\eta}] (m-1-2) %
      (m-1-2) edge[labelr={\mu}] (m-2-2) %
    }
  \end{center}
  commutes.
\end{lem}

\begin{cor}
  $(\L, \mu, \eta)$ is a monad on $\StdSig$.
\end{cor}

A crucial result is:
\begin{prop}
  For all   $\Gam \vdash Q \rcolon N \red N' \colon \Del$ and
  $\Del \vdash P \rcolon M \red M' \colon A$, we have:
$$\Gam \vdash P[Q] \req (M[Q] \vcomp{M[N']} P[N']) \rcolon M[N] \red M'[N'] \colon A.$$
\end{prop}
\proof
  We proceed by induction on $P$. Most cases are
  bureaucratic. Consider for instance $P = \constapp{c}{P_1, \ldots, P_p}$.
  Then, by definition: $$P[Q] = (\constapp{c}{P_1[N], \ldots, P_p[N]} 
  \vcomp{\constapp{c}{M'_1[N], \ldots, M'_p[N]}} \constapp{c}{M'_1[Q], \ldots, M'_p[Q]}.$$ By
  the third lifting rule, this is permutation equivalent to
  $$\constapp{c}{P_1[N] \vcomp{M'_1[N]} M'_1[Q], \ldots, P_p[N] \vcomp{M'_p[N]} M'_p[Q]}.$$
  By $p$ applications of the induction hypothesis, we obtain
  $$\constapp{c}{M_1[Q] \vcomp{M_1[N']} P_1[N'], \ldots, M_p[Q] \vcomp{M_p[N']} P_p[N']},$$
      which by lifting again yields the desired result:
  $$\constapp{c}{M_1[Q], \ldots, M_p[Q]} \vcomp{\constapp{c}{M_1[N'], \ldots, M_p[N']}} 
  \constapp{c}{P_1[N'], \ldots, P_p[N']}.$$

  The case where something actually happens is $P = \ruleapp{r}{P_1,
    \ldots, P_p}$, with $r \in X (G \vdash M_0, M'_0 \colon A)$ and
  each $\Del \vdash P_i \rcolon M_i \red M'_i \colon G_i$. Then, the
  left-hand side is $$\ruleapp{r}{P_1[N], \ldots, P_p[N]}
  \vcomp{M'_0[M'_1, \ldots, M'_n][N]} M'_0[M'_1, \ldots, M'_p][Q].$$
  By lifting, omitting indices of vertical compositions and implicitly
  using the \emph{category} rules, we have
  $$\ruleapp{r}{P_1[N], \ldots, P_p[N]} \req \ruleapp{r}{M_1[N], \ldots, M_p[N]} \vcomp{} 
  M'_0[P_1[N], \ldots, P_p[N]].$$ Observing that $M'_0[M'_1, \ldots,
  M'_p][Q] = M'_0[M'_1[Q], \ldots, M'_p[Q]]$, the whole is
  permutation equivalent to
  $$
  \begin{array}{ll}
    \ruleapp{r}{M_1[N], \ldots, M_p[N]} \vcomp{} \\
  M'_0[P_1[N], \ldots, P_p[N]] \vcomp{} \\
  M'_0[M'_1[Q], \ldots, M'_p[Q]],
  \end{array}
  $$
  i.e., by lifting (inductively):
  $$
  \begin{array}{ll}
    \ruleapp{r}{M_1[N], \ldots, M_p[N]} \vcomp{} \\
  M'_0[(P_1[N] \vcomp{} M'_1[Q]), \ldots, (P_p[N] \vcomp{} M'_p[Q])].
  \end{array}
  $$
  By induction hypothesis, this is permutation equivalent to
  $$
  \begin{array}{ll}
    \ruleapp{r}{M_1[N], \ldots, M_p[N]} \vcomp{} \\
  M'_0[(M_1[Q] \vcomp{} P_1[N']), \ldots, (M_p[Q] \vcomp{} P_p[N'])],
  \end{array}
  $$
  i.e., by lifting again to
  $$
  \begin{array}{ll}
    \ruleapp{r}{M_1[N], \ldots, M_p[N]} \vcomp{} \\
  M'_0[M_1[Q], \ldots, M_p[Q]]  \vcomp{} \\
  M'_0[P_1[N'], \ldots, P_p[N']].
  \end{array}
  $$
  The second lifting rule then yields
  $$
  \begin{array}{ll}
    \ruleapp{r}{M_1[Q], \ldots, M_p[Q]} \vcomp{} \\
  M'_0[P_1[N'], \ldots, P_p[N']],
  \end{array}
  $$
and hence
  $$
  \begin{array}{ll}
  M_0[M_1[Q], \ldots, M_p[Q]]  \vcomp{} \\
    \ruleapp{r}{M_1[N'], \ldots, M_p[N']} \vcomp{} \\
  M'_0[P_1[N'], \ldots, P_p[N']],
  \end{array}
  $$
  so, by the second lifting rule again:
  $$
  \begin{array}{ll}
  M_0[M_1[Q], \ldots, M_p[Q]]  \vcomp{} \\
    \ruleapp{r}{P_1[N'], \ldots, P_p[N']},
  \end{array}
  $$
  i.e., the right-hand side. \qed

\section{Cartesian closed 2-categories}\label{sec:tcccat}
\subsection{Definition}
In a 2-category $\C$, a diagram
$A \xleftarrow{p} C \xrightarrow{q} B$
is a \emph{product diagram} iff for all object $D$, the induced functor
$$\C (D,C)
 \xrightarrow{\pairing{\C (D,p), \C (D,q)}}  \C (D, A) \times \C (D, B)$$
is an isomorphism of categories.  Because this family of functors is 2-natural in
$D$, the inverse functors will also be 2-natural.

Similarly, an object $1$ of $\C$ is \emph{terminal} iff for all $D$
the unique functor
$$\C (D, 1) \xrightarrow{!} 1$$
is an isomorphism (where the right-hand $1$ is the terminal category).

\begin{defi}
  A \emph{2-category with finite products}, or fp 2-category, is a
  2-category $\C$, equipped with a terminal object and a 2-functor
  $$\C \times \C \xrightarrow{\times} \C,$$
  plus, for all $A$ and $B$, a product diagram
  $$A \xleftarrow{p} A \times B \xrightarrow{q} B.$$
\end{defi}

In such an fp 2-category $\C$, given objects $A$ and $B$, an
\emph{exponential} for them is a pair of an object $B^A$ and a
morphism $\ev \colon A \times B^A \to B$, such that for all $D$, the
functor
\begin{center}
  \diag(.6,.4){ %
    \& \C (A, A) \times \C (D, B^A) \& \& \& %
    |(UR)| \C (A \times D, A \times B^A) \& \\
    |[text height=1em]| \C (D, B^A) \& \& \& \& \& |[text height=1em](DR)| \C (A \times D, B) }{ (m-2-1)
    edge[labelal={(\id_A!, \id)}] (m-1-2) %
    (m-1-2) edge[labelu={\times}] (UR) %
    (UR) edge[labelar={\C (A \times D, \ev)}] (DR) %
  } %
\end{center}
is an isomorphism.  As above, because this family of functors is
2-natural in $D$, the inverse functors will also be 2-natural.

\begin{defi}
  A \emph{cartesian closed 2-category}, or \tccc{}, is an fp
  2-category, equipped with a choice of exponentials for all pairs of
  objects. The category \tcccat{} has \tcccs{} as objects, and stricly
  structure-preserving functors between them as morphisms.
\end{defi}

We observe in particular that this implies preservation of projections
and evaluation morphisms.

\section{Main adjunction}\label{sec:main}

\subsection{Right adjoint}
Given a \tccc{} $\C$, define $\V (\C) = (\C_0, \C_1, \linebreak \C_2)$
as follows.  First, let as in Section~\ref{sec:ccc} $(\C_0, \C_1) =
\Vi (\C)$, and recall the canonical $\Lo$ and $\Li$-algebra structures
$h_0$ and $h_1$.  Let then the reduction rules in $\C_2 (G \vdash M, N
\colon A)$ be the 2-cells in $\C (h_0 (G), h_0 (A)) (h_1 (M), h_1
(N))$, abbreviated to $\C (G, A) (M, N)$ in the sequel.

This signature $\V\C$ has a canonical $\L$-algebra structure $h_2
\colon \L(\V\C) \to \V\C$, which we define by induction over reductions in
Figure~\ref{fig:xi}.  In the case for $\lam$, $\varphi$ denotes the structure
isomorphism $\C ((\prod \Gam) \times A, B) \iso \C (\prod \Gam, B^A)$.
 
\begin{figure}[p]
  \begin{framed}
    \centering
  \[
\begin{array}[c]{l}
  (G \vdash x_i \rcolon x_i \red x_i \colon G_i )  \mapsto  
  (\id_{\pi_i} \colon \pi_i \to \pi_i \colon \prod G \to G_i) 
  \\[.2em] \hline \\[-.7em]
  (G \vdash () \rcolon () \red () \colon 1)  \mapsto 
  (\id_! \colon ! \to ! \colon \prod G \to 1) 
  \\[.2em] \hline \\[-.7em]
  (\Gam \vdash \constapp{c}{P_1, \ldots, P_n} \rcolon \constapp{c}{M_1, \ldots, M_n} \red 
  \constapp{c}{N_1, \ldots, N_n} \colon A)  \mapsto  
  \\
  \hspace*{.1cm}
  \diag(1,1){
    \prod \Gam \& \& \prod G  \& A
  }{
    (m-1-1) edge[bend left,twou={\pairing{M_1, \ldots, M_n}}] (m-1-3) %
    (m-1-1) edge[bend right,twod={\pairing{N_1, \ldots, N_n}}] (m-1-3) %
    (m-1-3) edge[labelu={c}] (m-1-4) %
    (u) edge[cell={-0.03},labelr={P}] (d) %
  }  \hspace*{.1cm} (c \in \C_1 (G, A), P = \pairing{P_1, \ldots, P_n})
  \\[.2em] \hline \\[-.7em]
  (\Gam \vdash \ruleapp{r}{P_1, \ldots, P_n} \rcolon M\subs{M_1, \ldots, M_n} \red 
  N\subs{N_1, \ldots, N_n} \colon A)  \mapsto 
  \\
  \hspace*{.1cm}
  \diag|baseline=(m-1-1.base)|{\prod \Gam \& \prod G \& A}{%
    (m-1-1) edge[bend left,twoup={up}{\pairing{M_1, \ldots, M_n}}] (m-1-2) %
    (m-1-1) edge[bend right,twodown={down}{\pairing{N_1, \ldots, N_n}}] (m-1-2) %
    (m-1-2) edge[bend left,twou={M}] (m-1-3) %
    (m-1-2) edge[bend right,twod={N}] (m-1-3) %
    (u) edge[cell={-0.03},labelr={r}] (d) %
    (up) edge[cell={-0.03},labelr={P}] (down) %
  }   \hspace*{.3cm} (P = \pairing{P_1, \ldots, P_n})
  \\[.2em] \hline \\[-.7em]
  (G \vdash P \vcomp{M_2} Q \rcolon M_1 \red M_3 \colon A) \mapsto 
  \diag{%
    |[text width=2em]| \prod G \& |[text width=2em]| A %
  }{%
    (m-1-1) edge[bend left=40, %
    postaction={ %
      decorate, decoration={ markings, mark=at position .5 with \node (Mi){};}}, %
    labelu={M_1}] (m-1-2) %
    (m-1-1) edge[two={Mii}{}] (m-1-2) %
    (m-1-1) edge[bend right=40, %
    postaction={ %
      decorate, decoration={ markings, mark=at position .5 with \node (Miii){};}}, %
    labeld={M_3}] (m-1-2) %
    (Mi) edge[cell={0},labelr={P}] (Mii) %
    (Mii) edge[cell={0},labelr={Q}] (Miii) %
  } \\[.2em] \hline \\[-.7em]
  (\Gam \vdash \lam x \colon A . P \rcolon \lam x \colon A. M \red
  \lam x \colon A. N \colon B^A) \mapsto 
  \varphi (P \colon M \red N \colon (\prod \Gam) \times A  \to B) 
  \\[.2em] \hline \\[-.7em]
  (\Gam \vdash P Q \rcolon M N \red M' N' \colon B) \mapsto 
  \diag(1,1){
    \prod \Gam \& \& B^A \times A \& B 
  }{
    (m-1-1) edge[bend left,twou={\pairing{M,N}}] (m-1-3) %
    (m-1-1) edge[bend right,twod={\pairing{M',N'}}] (m-1-3) %
    (m-1-3) edge[labelu={\ev}] (m-1-4) %
    (u) edge[twocell,labelr={{\scriptstyle \pairing{P,Q}}}] (d) %
  } \\[.2em] \hline \\[-.7em]
  (\Gam \vdash (P, Q) \rcolon (M, N) \red (M', N') \colon A \times B) \mapsto 
  \diag{
    \prod \Gam \& A \times B
  }{
    (m-1-1) edge[bend left,twou={\pairing{M,N}}] (m-1-2) %
    (m-1-1) edge[bend right,twod={\pairing{M',N'}}] (m-1-2) %
    (u) edge[twocell,labelr={{\scriptstyle \pairing{P,Q}}}] (d) %
  } \\[.2em] \hline \\[-.7em]
  (\Gam \vdash \proj_{A,B} P \rcolon \proj_{A,B} M \red \proj_{A,B} N \colon A) \mapsto 
  \diag(1,1){
    \prod \Gam \& \& A \times B \& A
  }{
    (m-1-1) edge[bend left,twou={M}] (m-1-3) %
    (m-1-1) edge[bend right,twod={N}] (m-1-3) %
    (m-1-3) edge[labelu={\proj}] (m-1-4) %
    (u) edge[twocell,labelr={P}] (d) %
  } \\[.2em] \hline \\[-.7em]
  (\Gam \vdash \projj_{A,B} P \rcolon \projj_{A,B} M \red \projj_{A,B} N \colon B) \mapsto 
  \diag(1,1){
    \prod \Gam \& \& A \times B \& B
  }{
    (m-1-1) edge[bend left,twou={M}] (m-1-3) %
    (m-1-1) edge[bend right,twod={N}] (m-1-3) %
    (m-1-3) edge[labelu={\projj}] (m-1-4) %
    (u) edge[twocell,labelr={P}] (d) %
  } 
\end{array}
\]
\end{framed}
\caption{The $\L$-algebra structure on $\V(\C)$}
\label{fig:xi}
\end{figure}

In order for the definition to make sense as a morphism $\L
(\V\C) \to \V\C$, we have to check its compatibility with the
equations. We have first:
\begin{lem}\label{lem:substsem}
  For all $\Del \vdash Q \rcolon N \red N' \colon \Gam$ and $\Gam
  \vdash P \rcolon M \red M' \colon A$ in $\L (\V\C)$,
  \begin{center}
    \diag(1,3.5){%
      \Del \& A %
    }{%
      (m-1-1) edge[bend left=20,twou={M[N]}] (m-1-2) %
      (m-1-1) edge[bend right=20,twod={M'[N']}] (m-1-2) %
      (u) edge[cell={0},labelr={\scriptstyle h_2 (P[Q])}] (d) %
    } = %
    \diag(1,2.5){%
      \Del  \& \Gam \& A. %
    }{%
      (m-1-1) edge[bend left,twou={N}] (m-1-2) %
      (m-1-1) edge[bend right,twod={N'}] (m-1-2) %
      (m-1-2) edge[bend left,twoup={up}{M}] (m-1-3) %
      (m-1-2) edge[bend right,twodown={down}{M'}] (m-1-3) %
      (u) edge[cell={0},labelr={\scriptstyle  h_2 (Q)}] (d) %
      (up) edge[cell={0},labelr={\scriptstyle  h_2 (P)}] (down) %
    }
  \end{center}
\end{lem}
\begin{proof} By induction on $P$ and the axioms for cartesian closed
  2-categories.
\end{proof}

\begin{lem}
  Any two equated reductions are mapped to the same 2-cell in $\C$.
\end{lem}
\begin{proof}
  We proceed by induction on the proof of the considered equation.
  The congruence and category rules of Figures~\ref{fig:req} and
  \ref{fig:req1} hold because, in $\C$, vertical composition is
  associative and unital, and equality is a congruence.  The beta rule
  is less easy, so we spell it out.

  The left-hand reduction is interpreted in $\C$ as 
  \begin{center}
      \diag(1,1.2){
    \prod \Gam \& \& B^A \times A \& B 
  }{
    (m-1-1) edge[bend left,twou={\pairing{\varphi M,N}}] (m-1-3) %
    (m-1-1) edge[bend right,twod={\pairing{\varphi M',N'}}] (m-1-3) %
    (m-1-3) edge[labelu={\ev}] (m-1-4) %
    (u) edge[cell={0},labelr={\scriptstyle \pairing{\varphi P,Q}}] (d) %
  }
  \end{center}
which is equal to 
  \begin{center}
      \diag (1,1.2) {
    \prod \Gam \& \& \prod \Gam \times A \& \& B^A \times A \& B 
  }{
    (m-1-1) edge[bend left,twou={\pairing{\id,N}}] (m-1-3) %
    (m-1-1) edge[bend right,twod={\pairing{\id,N'}}] (m-1-3) %
    (u) edge[twocell,labelr={\scriptstyle \pairing{\id,Q}}] (d) %
    (m-1-3) edge[bend left,twoup={up}{\varphi M \times A}] (m-1-5) %
    (m-1-3) edge[bend right,twodown={down}{\varphi M' \times A}] (m-1-5) %
    (m-1-5) edge[labelu={\ev}] (m-1-6) %
    (up) edge[twocell,labelr={\scriptstyle \varphi P \times A}] (down) %
  }
  \end{center}
which is in turn equal (by cartesian closedness of $\C$) to:
  \begin{center}
      \diag (1,1.2) {
    \prod \Gam \& \& \prod \Gam \times A \& \& B 
  }{
    (m-1-1) edge[bend left,twou={\pairing{\id,N}}] (m-1-3) %
    (m-1-1) edge[bend right,twod={\pairing{\id,N'}}] (m-1-3) %
    (u) edge[twocell,labelr={\scriptstyle \pairing{\id,Q}}] (d) %
    (m-1-3) edge[bend left,twoup={up}{M}] (m-1-5) %
    (m-1-3) edge[bend right,twodown={down}{M'}] (m-1-5) %
    (up) edge[twocell,labelr={\scriptstyle P}] (down) %
  }
  \end{center}
  and hence to the right-hand side of the equation by
  Lemma~\ref{lem:substsem}.  The other beta and eta rules similarly
  hold by the properties of products, internal homs, and terminal object in $\C$.

  The lifting rules hold by (particular cases of) the interchange law
  in $\C$ and functoriality of the structural isomorphisms
  \begin{center}
    $\C (A \times B, C) \iso \C (B, C^A)$ \hfil and \hfil
    $\C (C, A \times B) \iso \C (C, A) \times \C (C, B)$,
\end{center}
which concludes the proof.
\end{proof}

This assignment extends to cartesian closed functors and we have:
\begin{prop}
  $\V$ is a functor $\tcccat \to \StdSig$.
\end{prop}

\subsection{Left adjoint}\label{subsec:left}
Given an $\L$-algebra $h \colon \L (X) \to X$, we now construct a
\tccc{} $\F (X, h)$. It has:
\begin{iteMize}{$\bullet$}
\item objects the types in $\Lo (X_0)$;
\item 1-cells $A \to B$ the terms in $\Li (X_0, X_1) (A, B)$;
\item 2-cells $M \red N \colon A \to B$ the reduction rules in $X_2 (M,
  N)$.
\end{iteMize}
We then must define the \tccc{} structure, and we start with the
2-category structure. Composition of 1-cells $A \xrightarrow{M} B
\xrightarrow{N} C$ is defined to be $A \xrightarrow{N\subs{M}} C$.
Identities are given by variables, as usual.
Vertical composition of 2-cells
\begin{center}
  \Diagpetit{\path
    (u) edge[cell={-.1},labelr={\alpha}] (d) %
    (up) edge[cell={-.1},labelr={\beta}] (down) %
    ; %
}{ A \&  \& B %
  }{ %
    (m-1-1) edge[bend left,twou={M_1}] (m-1-3) %
    (m-1-1) edge[twod={},twoup={up}{}] node[above,pos=.3]{$M_2$} (m-1-3) %
    (m-1-1) edge[bend right,twodown={down}{M_3}] (m-1-3) %
  }
\end{center}
is given by $h (\eta(\alpha) \vcomp{M_2} \eta(\beta))$.

Horizontal composition of 2-cells
\begin{equation}
  \diag{ A  \& B \& C %
  }{ %
    (m-1-1) edge[bend left,twou={M}] (m-1-2) %
    (m-1-1) edge[bend right,twod={M'}] (m-1-2) %
    (m-1-2) edge[bend left,twoup={up}{N}] (m-1-3) %
    (m-1-2) edge[bend right,twodown={down}{N'}] (m-1-3) %
    (u) edge[cell={0},labelr={\alpha}] (d) %
    (up) edge[cell={0},labelr={\beta}] (down) %
  }
  \label{eq:cells}
\end{equation}
is obtained as $h (\ruleapp{\beta}{\eta(\alpha)})$.

The identity at $A \xrightarrow{M} B$ is $h(M)$.

To show that this yields a 2-category structure, the only non obvious
point is the interchange law. We deal with it using the following
series of results.  First, consider the \emph{left whiskering}
  \begin{center}
    \diag{ A  \& B \& C %
    }{ %
      (m-1-1) edge[bend left,twou={M}] (m-1-2) %
      (m-1-1) edge[bend right,twod={M'}] (m-1-2) %
      (m-1-2) edge[labelu={N}] (m-1-3) %
      (u) edge[cell={0},labelr={\alpha}] (d) %
    }
  \end{center}
  of a 2-cell $\alpha$ by a 1-cell $N$, i.e., the composition $\id_N
  \rond \alpha = h (\ruleapp{(h (N))}{\eta(\alpha)})$. 
  \begin{lem} We have: $h (\ruleapp{(h (N))}{\eta(\alpha)}) = h (N\subs{\eta(\alpha)})$.
  \end{lem}
  \begin{proof}
    Indeed, consider the term $\ruleapp{N}{\eta (\eta (\alpha))}$ in $\L(\L (X))$. Its
    images by $h \rond \L (h)$ and $h \rond \mu$ coincide, and are
    respectively $h (\ruleapp{(h (N))}{\eta(\alpha)})$, i.e., $\id_N \rond \alpha$,
    and $h (N\subs{\eta(\alpha)})$.
  \end{proof}
  Similarly, consider the \emph{right whiskering}
  \begin{center}
    \diag{ A  \& B \& C %
    }{ %
      (m-1-2) edge[bend left,twou={N}] (m-1-3) %
      (m-1-2) edge[bend right,twod={N'}] (m-1-3) %
      (m-1-1) edge[labelu={M}] (m-1-2) %
      (u) edge[cell={0},labelr={\gamma}] (d) %
    }
  \end{center}
  of a 2-cell $\gamma$ by a 1-cell $M$, i.e., the composition $\gamma
  \rond \id_N = h (\ruleapp{\gamma}{\eta(h (M))})$.
  \begin{lem}
    We have: $h (\ruleapp{\gamma}{\eta(h (M))}) = h (\ruleapp{\gamma}{M})$.
  \end{lem}
  \begin{proof}
    Consider $\ruleapp{(\eta \gamma)}{\eta M}$ in $\L (\L (X))$. Its images by
    $h \rond \L (h)$ and $h \rond \mu$ coincide, and are respectively
    $h (\ruleapp{\gamma}{\eta (h (M))})$ and $h (\ruleapp{\gamma}{M})$.
  \end{proof}
  Now, we prove that the two sensible ways of mimicking horizontal composition using 
  whiskering coincide with actual horizontal composition:
  \begin{lem}
    For any cells as in~\eqref{eq:cells}, 
    $$(\beta \rond \id_M) \vcomp{} (\id_{N'} \rond \alpha) =
    \beta \rond \alpha = (\id_N \rond \alpha) \vcomp{} (\beta \rond \id_{M'}).$$
  \end{lem}
  \begin{proof}
    Consider first the reduction $\eta (\ruleapp{\beta}{M}) \vcomp{} \eta
    (N'\subs{\eta(\alpha)})$ in $\L (\L (X))$.  Taking $h \rond \L
    (h)$ and $h \rond \mu$ as above respectively yields
    \begin{iteMize}{$\bullet$}
    \item $h (\eta (h (\ruleapp{\beta}{M})) \vcomp{} \eta (h
      (N'\subs{\eta(\alpha)})))$, and
    \item $h (\ruleapp{\beta}{M} \vcomp{} N'\subs{\eta(\alpha)}) = h (\ruleapp{\beta}{\eta(\alpha)})$,
    \end{iteMize}
    hence the left-hand equality. Then consider $\eta (N\subs{\eta
      (\alpha)}) \vcomp{} \eta (\beta (M'))$. Evaluating as before
    yields the right-hand equality.
  \end{proof}
  Furthermore, consider any configuration like:
  \begin{center}
    \diagpetit{ A \&  \& B \& C. %
    }{ %
      (m-1-1) edge[bend left,twou={M}] (m-1-3) %
      (m-1-1) edge[twod={},twoup={up}{}] node[above,pos=.3]{$M'$}  (m-1-3) %
      (m-1-1) edge[bend right,twodown={down}{M''}] (m-1-3) %
      (m-1-3) edge[labelu={N}] (m-1-4) %
      (u) edge[cell={-.1},labelr={\alpha}] (d) %
      (up) edge[cell={-.1},labelr={\beta}] (down) %
    }
  \end{center}
  \begin{lem}
    We have $(\id_N \rond \alpha) \vcomp{} (\id_N \rond \beta) = 
    \id_N \rond (\alpha \vcomp{} \beta)$.
  \end{lem}
  \proof
    Consider $\eta (N\subs{\eta (\alpha)}) \vcomp{} \eta (N\subs{\eta (\beta)})$.
    Evaluating yields equality of
    \begin{iteMize}{$\bullet$}
    \item $h (\eta (h(N\subs{\eta (\alpha)})) \vcomp{} \eta (h(N\subs{\eta (\beta)})))$, i.e.,
      the left-hand side, and
    \item $h (N\subs{\eta (\alpha)} \vcomp{} N\subs{\eta (\beta)})$,
      i.e., $h (N\subs{\eta (\alpha) \vcomp{} \eta (\beta)})$ by
      lifting.
    \end{iteMize}
    But now consider $N\subs{\eta (\eta (\alpha) \vcomp{} \eta
      (\beta))}$. Evaluating yields equality of
    \begin{iteMize}{$\bullet$}
    \item $h (N\subs{\eta (\alpha) \vcomp{} \eta (\beta)})$, as above, and
    \item $h (N\subs{\eta (h (\eta (\alpha) \vcomp{} \eta
        (\beta)))})$, i.e., $h (N\subs{\eta (\alpha \vcomp{} \beta)})$
      (where $\alpha \vcomp{} \beta$ denotes vertical composition in
      our candidate 2-category), i.e., the right-hand side.   \qed
    \end{iteMize}

Finally, by a similar argument, we have:
  \begin{lem}
For any
  \begin{center}
    \diagpetit{ A \&  B \& \& C, %
    }{ %
      (m-1-2) edge[bend left,twou={N}] (m-1-4) %
      (m-1-2) edge[twod={},twoup={up}{}] node[above,pos=.3]{$N'$}  (m-1-4) %
      (m-1-2) edge[bend right,twodown={down}{N''}] (m-1-4) %
      (m-1-1) edge[labelu={M}] (m-1-2) %
      (u) edge[cell={-.1},labelr={\alpha}] (d) %
      (up) edge[cell={-.1},labelr={\beta}] (down) %
    }
  \end{center}
    we have $(\alpha \rond \id_M) \vcomp{} (\beta \rond \id_M) =
    (\alpha \vcomp{} \beta) \rond \id_M$.  \end{lem}

  \begin{lem}
    The interchange law holds, i.e., for all reduction rules as in
    \begin{center}
  \Diagpetit{\path
    (u) edge[cell={-.1},labelr={\alpha}] (d) %
    (up) edge[cell={-.1},labelr={\beta}] (down) %
    (ui) edge[cell={-.1},labelr={\gamma}] (di) %
    (upi) edge[cell={-.1},labelr={\theta}] (downi) %
    ; %
}{ A \&  \& B  \& \& C,%
  }{ %
    (m-1-1) edge[bend left,twou={M_1}] (m-1-3) %
    (m-1-1) edge[twod={},twoup={up}{}] node[above,pos=.3]{$M_2$} (m-1-3) %
    (m-1-1) edge[bend right,twodown={down}{M_3}] (m-1-3) %
    (m-1-3) edge[bend left,twoup={ui}{N_1}] (m-1-5) %
    (m-1-3) edge[twodown={di}{},twoup={upi}{}] node[above,pos=.3]{$N_2$} (m-1-5) %
    (m-1-3) edge[bend right,twodown={downi}{N_3}] (m-1-5) %
  }
\end{center}
we have $$(\gamma \vcomp{} \theta) \rond (\alpha \vcomp{} \beta)
 = (\gamma \rond \alpha) \vcomp{} (\theta \rond \beta).$$
  \end{lem}
\begin{proof}
By the previous results, we have
$$
\begin{array}{l}
  (\gamma \vcomp{} \theta) \rond (\alpha \vcomp{} \beta) \\
  {} = ((\gamma \vcomp{} \theta) \rond M_1) \vcomp{}
  (N_3 \rond (\alpha \vcomp{} \beta)) \\
  {} = (\gamma \rond M_1) \vcomp{} (\theta \rond M_1) \vcomp{}
  (N_3 \rond \alpha) \vcomp{} (N_3 \rond \beta) \\
  {} = (\gamma \rond M_1) \vcomp{} (N_2 \rond \alpha) \vcomp{}
  (\theta \rond M_2) \vcomp{} (N_3 \rond \beta) \\
  {} = (\gamma \rond \alpha) \vcomp{} (\theta \rond \beta). 
\end{array}
$$
\end{proof}

Now, let us show cartesian closedness. 
We have a bijection of hom-sets $\Li (X) (C \vdash A \times B)
\iso \Li (X) (C \vdash A) \times \Li (X) (C \vdash B)$, given by 
$$
\begin{array}{rcll}
\Li (X) (C \vdash A \times B)
& \to & \Li (X) (C \vdash A) \times \Li (X) (C \vdash B) \\
M & \mapsto & \proj{M}, \projj{M}
\end{array}
$$ 
and 
$$
\begin{array}{rcll}
\Li (X) (C \vdash A) \times \Li (X) (C \vdash B) & \to &
\Li (X) (C \vdash A \times B) \\
M, N & \mapsto & (M, N).
\end{array}
$$
These are mutually inverse thanks to the beta and eta rules for
products in the simply-typed $\lam$-calculus.

On 2-hom-sets, we have 
$$\hspace*{-.2cm}
\begin{array}{rcll}
\L (X) (C \vdash M, N \colon A \times B)
& \to &
\L (X) (C \vdash \proj{M}, \proj{N} \colon A)
\times 
\L (X) (C \vdash \projj{M}, \projj{N} \colon B) \\
P & \mapsto & \proj{P}, \projj{P}
\end{array}$$
and (omitting $C$)
$$\begin{array}{rcll}
\L (X) (M_1, N_1 \colon A)
\times 
\L (X) (M_2, N_2 \colon B) 
& \to &
\L (X) ((M_1, M_2), (N_1, N_2) \colon A \times B) \\
P_1, P_2 & \mapsto & (P_1, P_2),
\end{array}$$
which are mutually inverse thanks to the beta and eta rules for products
in Figure~\ref{fig:req1}. We use these to define the desired 
isomorphism  $(u, v)$
$$X_2(C \vdash M, N \colon A \times B) \iso
X_2 (C \vdash \proj{M}, \proj{N} \colon A) \times 
X_2 (C \vdash \projj{M}, \projj{N} \colon B),$$
as in the diagrams
\begin{center}
  \diag{%
    X_2 (M, N) \& %
    X_2 (\proj{M}, \proj{N}) \times X_2 (\projj{M}, \projj{N}) \\
    \L (X) (M, N) \& %
    \L (X) (\proj{M}, \proj{N}) \times \L (X) (\projj{M}, \projj{N})
  }{%
    (m-1-1) edge[labelu={u}] (m-1-2) %
    edge[labell={\eta}] (m-2-1) %
    (m-2-1) edge[labeld={\iso}] (m-2-2) %
    (m-2-2) edge[labelr={h \times h}] (m-1-2) %
  }
\end{center}
and 
\begin{center}
  \diag{%
    X_2 (\proj{M}, \proj{N}) \times X_2 (\projj{M}, \projj{N})\&  %
    X_2 (M, N) \\
    \L (X) (\proj{M}, \proj{N}) \times \L (X) (\projj{M}, \projj{N}) \& %
    \L (X) (M, N). %
  }{%
    (m-1-1) edge[labelu={v}] (m-1-2) %
    edge[labell={\eta \times \eta}] (m-2-1) %
    (m-2-1) edge[labeld={\iso}] (m-2-2) %
    (m-2-2) edge[labelr={h}] (m-1-2) %
  }
\end{center}

Starting from $r \in X_2 (M, N)$, we
obtain $$v (u (r)) = h (\eta (h (\pi (\eta (r)))), \eta (h (\pi' (\eta
(r))))).$$ But consider $(\eta (
\pi \eta (r)), \eta ( \pi' \eta (r)))$ in $\L (\L X)$; its
images by $h \rond \L h$ and $h \rond \mu$ are respectively:
\begin{iteMize}{$\bullet$}
\item $h (\eta (h (\pi (\eta r))), \eta (h (\pi' (\eta r))))$, and
\item $h (\pi \eta(r), \pi' \eta(r))$, i.e., $h (\eta (r))$, i.e., $r$, 
\end{iteMize}
which must be equal because $h$ is an $\L$-algebra, hence $v \rond u =
\id$.

Conversely, starting from $(r, s) \in X_2 (M_1, M_2) \times
X_2(N_1, N_2)$, we obtain the pair with components 
\begin{center}
  $h(\pi(\eta (h (\eta (r), \eta (s)))))$ \hfil and \hfil
  $h(\pi'(\eta (h (\eta (r), \eta (s)))))$.
\end{center}
Considering $\pi (\eta (\eta (r), \eta (s))) \in \L (\L (X))$,
its images by $h \rond \L (h)$ and $h \rond \mu$ are respectively:
\begin{iteMize}{$\bullet$}
\item $h(\pi (\eta (h(\eta (r), \eta (s)))))$, and
\item $h(\pi (\eta (r), \eta (s))) = h (\eta (r)) = r$.
\end{iteMize}
As above, they must be equal, and by symmetry the second component is
$s$, and we have proved $u \rond v = \id$. Similar reasoning for
the terminal object and internal homs leads to:
\begin{prop}
  This yields a \tccc{} structure on $\C$.
\end{prop}

This extends to morphisms of $\L$-algebras, so we have constructed a
functor $\F \colon \LAlg \to \tcccat$.

\subsection{Adjunction}
Consider any $\L$-algebra $(X, h)$. What does $(Y, k) = \V (\F (X,
h))$ look like?  Sorts in $Y_0$ are types in $\Lo (X_0)$. Operations
$Y_1(G \vdash A)$ are terms in $\Li (X_0, X_1) (\mu(\prod G) \vdash
\mu(A))$, where $\mu$ denotes the monad multiplication for
$\Lo$. Reduction rules in $Y_2 (G \vdash M, N \colon B)$ are
reduction rules in $X_2 (\mu(\prod G) \vdash M', N' \colon \mu(B)),$ where $M' =
M\subs{\pi_1 x / x_1, \ldots, \pi_n x / x_n}$ (and similarly for
$N'$).

Let $\unitadj_{X}$ map:
\begin{iteMize}{$\bullet$}
\item each sort $\iota \in X_0$ to the type $\iota \in \Lo (X_0)$
  (this is the monad unit for $\Lo$),
\item each operation $c \in X (G \vdash A)$ to the term
  $\constapp{c}{\pi_1 x, \ldots, \pi_n x}$, and
\item each reduction rule $r \in X_2(G \vdash M, N \colon A)$ to
  the reduction rule
  $$h(\ruleapp{r}{\pi_1 x, \ldots, \pi_n x}) \in X_2(\prod G \vdash M', N' \colon A).$$
  (Thanks to the fact that $\mu (\Lo (\eta) (A)) = A$.)
\end{iteMize}

\begin{thm}
  This $\unitadj$ is a natural transformation which is the unit of an adjunction
  \begin{center}
    \adj{\LAlg}{\tcccat{}.}{\F}{\V}
  \end{center}
\end{thm}
\begin{proof}
  Consider any morphism $f \colon (X, h) \to \V (\C)$, and let $\X = \F (X,h)$, $(Y, k)
  = \V (\X)$, and $\V (\C) = (\C_0, \C_1, h_2 \colon \C_2 \to
  \C_1)$.  We now define a uniquely determined cartesian closed
  functor $f' \colon \X \to \C$ making the triangle
  \begin{center}
    \diag{%
      X \& \V (\X) \\
      \& \V (\C)
      }{%
        (m-1-1) edge[labelu={\unitadj_X}] (m-1-2) %
        edge[labelbl={f}] (m-2-2) %
        (m-1-2) edge[labelr={\V (f')}] (m-2-2) %
      }
  \end{center}
  commute.

  On objects, it is determined by induction: on sorts by $f_0$, and on
  type constructors by the requirement that $f'$ be cartesian
  closed. On morphisms, it is similarly determined by $f_1$ and $f'$
  being cartesian closed.  On 2-cells, define $f'$ to be $f_2 \colon
  X_2 (A \vdash M, N \colon B) \to \C (f' (A), f' (B)) (f' (M), f'
  (N))$, which is also the only possible choice from $f$.

  This indeed makes the above triangle commute, because any $r \in
  X_2(G \vdash M,N \colon A)$ is first mapped to $h (\ruleapp{r}{\pi_1
    x, \ldots, \pi_n x}) \in X_2(\prod G \vdash M',N' \colon A)$, and
  then to $$h (\ruleapp{r}{\pi_1 x, \ldots, \pi_n x})$$ in $\C$, which,
  because $f$ is a morphism of $\L$-algebras, is equal to
  $$h_2(\ruleapp{(f_2(r))}{\pi_1 x, \ldots, \pi_n x}),$$ i.e., to
  $f_2(r)$.

  It thus remains to show that $f'$ is cartesian closed, which
  follows by $f$ being a morphism of $\L$-algebras. For example, to
  show that binary pairings of reductions are preserved, consider $r
  \in X_2 (C \vdash M_1, M_2 \colon A)$ and $s \in X_2 (C \vdash N_1,
  N_2 \colon B)$. Their product in $\F (X)$ is obtained by considering
  the atomic reductions $x \colon C \vdash \ruleapp{r}{x} \rcolon M_1
  \red M_2 \colon A$ and $x \colon C \vdash \ruleapp{s}{x} \rcolon N_1
  \red N_2 \colon B$ and taking $h (\ruleapp{r}{x}, \ruleapp{s}{x})$,
  which is mapped by $f_2$ to $f_2 (h (\ruleapp{r}{x},
  \ruleapp{s}{x}))$.  But, because $f$ is a morphism of $\L$-algebras,
  this is the same as $h_2 (\ruleapp{(f_2 (r))}{x}, \ruleapp{(f_2
    (s))}{x})$, which is by definition (i.e., Figure~\ref{fig:xi}) the
  pairing $(f_2 (r) , f_2 (s))$ in $\C$.
\end{proof}

\subsection*{Acknowledgements} Thanks to Nicolas Tabareau for useful
feedback, and to Aurore Alcolei for reviving the subect after a few
years.

\bibliographystyle{plain}


\begin{thebibliography}{10}

\bibitem{bruggink:phd2008}
{H. J. Sander} Bruggink.
\newblock {\em Equivalence of reductions in higher-order rewriting}.
\newblock PhD thesis, Utrecht University, 2008.

\bibitem{Capriotti}
Paolo Capriotti.
\newblock Concurrent semantics with variable binding.
\newblock Master's thesis, University of Pisa, 2009.

\bibitem{DBLP:conf/rta/CorradiniGM95}
Andrea Corradini, Fabio Gadducci, and Ugo Montanari.
\newblock Relating two categorial models of term rewriting.
\newblock In Jieh Hsiang, editor, {\em RTA}, volume 914 of {\em Lecture Notes
  in Computer Science}, pages 225--240. Springer, 1995.

\bibitem{CROct}
R.~L. Crole.
\newblock {\em {Categories for Types}}.
\newblock Cambridge Mathematical Textbooks. Cambridge University Press, 1993.
\newblock xvii+335 pages, ISBN 0521450926HB, 0521457017PB.

\bibitem{Despeyroux}
Thierry Despeyroux and Andr{\'e} Hirschowitz.
\newblock Principles for functional abstract syntax.
\newblock Draft, 1995.

\bibitem{Engberg86acalculus}
Uffe Engberg and Mogens Nielsen.
\newblock A calculus of communicating systems with label passing.
\newblock Technical Report PB-208, Aarhus University, 1986.

\bibitem{DBLP:conf/lics/Fiore08}
Marcelo~P. Fiore.
\newblock Second-order and dependently-sorted abstract syntax.
\newblock In {\em LICS~'08}, pages 57--68. IEEE Computer Society, 2008.

\bibitem{DBLP:journals/tcs/Hilken96}
Barney~P. Hilken.
\newblock Towards a proof theory of rewriting: The simply typed
  2{$\lambda$}-calculus.
\newblock {\em Theor. Comput. Sci.}, 170(1-2):407--444, 1996.

\bibitem{hirscho:lam}
Andr{\'e} Hirschowitz and Marco Maggesi.
\newblock Modules over monads and linearity.
\newblock In Daniel Leivant and Ruy J. G.~B. de~Queiroz, editors, {\em WoLLIC},
  volume 4576 of {\em Lecture Notes in Computer Science}, pages 218--237.
  Springer, 2007.

\bibitem{KPT:sketches}
Yoshiki Kinoshita, John Power, and Makoto Takeyama.
\newblock Sketches.
\newblock {\em Journal of Pure and Applied Algebra}, 143(1-3), 1999.

\bibitem{KlopCRS}
Jan~W. Klop.
\newblock {\em Combinatory Reduction Systems}.
\newblock PhD thesis, CWI, Amsterdam, 1980.

\bibitem{Lambek:categorical}
Joachim Lambek and {Philip J.} Scott.
\newblock {\em Introduction to Higher-Order Categorical Logic}.
\newblock Cambridge University Press, 1986.

\bibitem{Milner80}
Robin Milner.
\newblock {\em A Calculus of Communicating Systems}, volume~92 of {\em LNCS}.
\newblock Springer, 1980.

\bibitem{Milner:pi}
Robin Milner, Joachim Parrow, and David Walker.
\newblock A calculus of mobile processes, {I/II}.
\newblock {\em Information and Computation}, 100(1):1--77, 1992.

\bibitem{DBLP:conf/lics/Nipkow91}
Tobias Nipkow.
\newblock Higher-order critical pairs.
\newblock In {\em LICS~'91}, pages 342--349. IEEE Computer Society, 1991.

\bibitem{terese}
TeReSe, editor.
\newblock {\em Term Rewriting Systems}.
\newblock Number~55 in Cambridge Tracts in Theoretical Computer Science.
  Cambridge University Press, 2003.

\bibitem{DBLP:conf/hoa/OostromR93}
Vincent van Oostrom and Femke van Raamsdonk.
\newblock Comparing combinatory reduction systems and higher-order rewrite
  systems.
\newblock In Jan Heering, Karl Meinke, Bernhard M{\"o}ller, and Tobias Nipkow,
  editors, {\em HOA}, volume 816 of {\em Lecture Notes in Computer Science},
  pages 276--304. Springer, 1993.

\bibitem{DBLP:journals/tcs/Wells90}
Charles Wells.
\newblock A generalization of the concept of sketch.
\newblock {\em Theor. Comput. Sci.}, 70(1):159--178, 1990.

\bibitem{Wolfram}
David~A. Wolfram.
\newblock {\em The Clausal Theory of Types}.
\newblock Number~21 in Cambridge Tracts in Theoretical Computer Science.
  Cambridge University Press, 1993.

\bibitem{Zsido}
Julianna Zsid{\'o}.
\newblock {\em Typed Abstract Syntax}.
\newblock PhD thesis, Universit{\'e} de Nice-Sophia Antipolis, 2010.

\end{thebibliography}

\newpage
\appendix
\section{Equations on reductions}\label{sec:req}

\begin{figure}[hb]
  \begin{framed}
    \centering \ruleset{Congruence}{ \inferrule{ \Gam \vdash P \rcolon
        M \red N \colon A }{ \Gam \vdash P \req P \rcolon M \red N
        \colon A } %
      \and
      \inferrule{%
        \Gam \vdash P \req Q \rcolon M \red N \colon A }{ %
        \Gam \vdash Q \req P \rcolon M \red N \colon A %
      } %
      \and
      \inferrule{%
        \Gam \vdash P_1 \req P_2 \rcolon M \red N \colon A \\
        \Gam \vdash P_2 \req P_3 \rcolon M \red N \colon A }{ %
        \Gam \vdash P_1 \req P_3 \rcolon M \red N \colon A }\and
      \inferrule{ %
        \Gam \vdash P \req P' \rcolon M_1 \red M_2 \colon A \\
        \Gam \vdash Q \req Q' \rcolon M_2 \red M_3 \colon A %
      }{%
        \Gam \vdash (P \vcomp{M_2} Q) \req (P' \vcomp{M_2} Q')
        \rcolon %
        M_1 \red M_3 \colon A %
      }%
      \and %
      \inferrule{ %
        (r \in X (G \vdash M, N \colon A)) \\
        \Gam \vdash P_1 \req Q_1 \rcolon M_1 \red N_1 \colon G_1 \\ \ldots \\
        \Gam \vdash P_n \req Q_n \rcolon M_n \red N_n \colon G_n }{ %
        \Gam \vdash \ruleapp{r}{P_1, \ldots, P_n} \req \ruleapp{r}{Q_1, \ldots, Q_n}
        \rcolon %
        M \subs{M_1, \ldots, M_n} \red N \subs{N_1, \ldots, N_n}
        \colon A %
      }\and %
      \inferrule{ %
        (c \in X (G \vdash A)) \\
        \Gam \vdash P_1 \req Q_1 \rcolon M_1 \red N_1 \colon G_1 \\ \ldots \\
        \Gam \vdash P_n \req Q_n \rcolon M_n \red N_n \colon G_n \\
      }{ %
        \Gam \vdash \constapp{c}{P_1, \ldots, P_n} \req \constapp{c}{Q_1, \ldots, Q_n}
        \rcolon %
        \constapp{c}{M_1, \ldots, M_n} \red \constapp{c}{N_1, \ldots, N_n} \colon A %
      } \and
      \inferrule{ %
        \Gam, x \colon A \vdash P \req Q \rcolon M \red N \colon B %
      }{%
        \Gam \vdash (\lam x \colon A . P) \req (\lam x \colon A . Q)
        \rcolon %
        \lam x \colon A. M \red \lam x \colon A. N \colon B^A %
      }\and
      \inferrule{ %
        \Gam \vdash P \req P' \rcolon M \red M' \colon B^A \\ %
        \Gam \vdash Q \req Q' \rcolon N \red N' \colon A %
      }{ %
        \Gam \vdash (P Q) \req (P' Q') \rcolon M N \red M' N' \colon B
      } %
      \and %
      \inferrule{%
        \Gam \vdash P \req P' \rcolon M \red M' \colon A \\ %
        \Gam \vdash Q \req Q' \rcolon N \red N' \colon B %
      }{ %
        \Gam \vdash (P, Q) \req (P', Q') \rcolon (M, N) \red (M', N')
        \colon A \times B } %
      \and %
      \inferrule{ %
        \Gam \vdash P \req Q \rcolon M \red N \colon A \times B %
      }{ %
        \Gam \vdash (\proj_{A,B} P) \req (\proj_{A,B} Q) \rcolon %
        \proj_{A,B} M \red \proj_{A,B} N \colon A %
      } %
      \and
      \inferrule{ %
        \Gam \vdash P \req Q \rcolon M \red N \colon A \times B %
      }{ %
        \Gam \vdash (\projj_{A,B} P) \req (\projj_{A,B} Q) \rcolon %
        \projj_{A,B} M \red \projj_{A,B} N \colon A %
      } %
    }
  \end{framed}
  \caption{Equations on reductions (Congruence)}
  \label{fig:req}
\end{figure}
\newpage

\begin{figure}[ht]
  \begin{framed}
    \centering 
    \ruleset{Category}{ \inferrule{ %
        \Gam \vdash P_1 \rcolon M_1 \red M_2 \colon A \\ %
        \Gam \vdash P_2 \rcolon M_2 \red M_3 \colon A \\ %
        \Gam \vdash P_3 \rcolon M_3 \red M_4 \colon A \\ %
      }{ %
        \Gam \vdash (P_1 \vcomp{M_2} (P_2 \vcomp{M_3} P_3)) %
        \req %
        ((P_1 \vcomp{M_2} P_2) \vcomp{M_3} P_3) %
        \rcolon M_1 \red M_4 \colon A \\ %
      } \and
      \inferrule{ %
        \Gam \vdash P \rcolon M \red N \colon A %
      }{ %
        \Gam \vdash (P \vcomp{N} N) \req P \rcolon M \red N \colon A %
      } \and
      \inferrule{ %
        \Gam \vdash P \rcolon M \red N \colon A %
      }{ %
        \Gam \vdash (M \vcomp{M} P) \req P \rcolon M \red N \colon A %
      }}\hrulefill
      \ruleset{Beta and Eta}{ %
      \inferrule{ %
        \Gam, x \colon A \vdash P \rcolon M \red M' \colon B \\
        \Gam \vdash Q \rcolon N \red N' \colon A
      }{ %
        \Gam \vdash ((\lam x \colon A. P) Q) \req P\subs{Q/x} \rcolon
        (\lam x \colon A. M) N \red M'\subs{N'/x} \colon B %
      } \and
      \inferrule{%
        \Gam \vdash P \rcolon M \red N \colon B^A }{%
        \Gam \vdash P \req \lam x \colon A. (P x) \rcolon M \red N
        \colon B^A }~(x \notin \Gam) \and %
      \inferrule{%
        \Gam \vdash P \rcolon M_1 \red M_2 \colon A \\
        \Gam \vdash Q \rcolon N_1 \red N_2 \colon B \\
      }{%
        \Gam \vdash \proj (P, Q) \req P \rcolon \proj (M_1, N_1) \red
        M_2 \colon A } \and %
      \inferrule{%
        \Gam \vdash P \rcolon M_1 \red M_2 \colon A \\
        \Gam \vdash Q \rcolon N_1 \red N_2 \colon B \\
      }{%
        \Gam \vdash \projj (P, Q) \req Q \rcolon \projj (M_1, N_1)
        \red N_2 \colon A } \and %
      \inferrule{%
        \Gam \vdash P \rcolon (M_1, N_1) \red (M_2, N_2) \colon A
        \times B }{%
        \Gam \vdash P \req (\proj P, \projj P) \rcolon (M_1, N_1) \red
        (M_2, N_2) \colon A \times B } \and \inferrule{\Gam \vdash P
        \rcolon M \red N \colon 1}{ %
        \Gam \vdash P \req () \rcolon M \red N \colon 1} }
  \end{framed}
  \caption{Equations on reductions (Category and Beta-Eta)}
  \label{fig:req1}
\end{figure}
\newpage 

\begin{figure}[ht]
  \begin{framed}
    \centering 
    \ruleset{Lifting}{ %
      \inferrule{ %
        (r \in X (\Gam \vdash M_1,M_2 \colon A))  \\
        \Del \vdash P \rcolon N_1 \red N_2 \colon \Gam \\
        \Del \vdash Q \rcolon N_2 \red N_3 \colon \Gam
      }{ %
        \Gam \vdash \ruleapp{r}{P \vcomp{N_2} Q} \req M_1\subs{P}
        \vcomp{M_1\subs{N_2}} \ruleapp{r}{Q} 
        \rcolon 
        M_1\subs{N_1} \red M_2\subs{N_3} \colon A
          %
      }%
      \and
      \inferrule{ %
        (r \in X (\Gam \vdash M_1,M_2 \colon A))  \\
        \Del \vdash P \rcolon N_1 \red N_2 \colon \Gam \\
        \Del \vdash Q \rcolon N_2 \red N_3 \colon \Gam
      }{ %
        \Gam \vdash \ruleapp{r}{P \vcomp{N_2} Q} \req \ruleapp{r}{P}
        \vcomp{M_2\subs{N_2}} M_2 \subs{Q} 
        \rcolon 
        M_1\subs{N_1} \red M_2\subs{N_3} \colon A
          %
      }%
      \and
      \inferrule{ %
        \Gam \vdash P \rcolon M_1 \red M_2 \colon G \\
        \Gam \vdash Q \rcolon M_2 \red M_3 \colon G %
      }{ %
        \Gam \vdash (c (P \vcomp{M_2} Q)) \req (c (P) \vcomp{c (M_2)}
        c (Q)) \rcolon M_1 \red M_3 \colon A }~(c \in X (G \vdash A))
      \and
      \inferrule{
        \Gam, x \colon A \vdash P \rcolon M_1 \red M_2  \colon B \\
        \Gam, x \colon A \vdash Q \rcolon M_2 \red M_3 \colon B %
      }{ \mbox{\begin{minipage}[t]{0.8\linewidth} \raggedright $ \Gam
            \vdash (\lam x \colon A. (P \vcomp{M_2} Q)) \req ((\lam x
            \colon A . P) \vcomp{\lam x \colon A. M_2} (\lam x \colon
            A. Q))$
            \\
            \raggedleft $\rcolon \lam x \colon A. M_1 \red \lam x
            \colon A. M_3 \colon B^A$
          \end{minipage}} } \and %
      \inferrule{
        \Gam \vdash P \rcolon M_1 \red M_2 \colon B^A \\
        \Gam \vdash P' \rcolon M_2 \red M_3 \colon B^A \\
        \Gam \vdash Q \rcolon N_1 \red N_2 \colon A \\
        \Gam \vdash Q' \rcolon N_2 \red N_3 \colon A }{ \Gam \vdash
        ((P \vcomp{M_2} P') (Q \vcomp{N_2} Q')) \req ((P Q) \vcomp{M_2
          N_2} (P' Q')) \rcolon M_1 N_1 \red M_3 N_3 \colon B } \and %
      \inferrule{
        \Gam \vdash P \rcolon M_1 \red M_2 \colon A \\
        \Gam \vdash P' \rcolon M_2 \red M_3 \colon A \\
        \Gam \vdash Q \rcolon N_1 \red N_2 \colon B \\
        \Gam \vdash Q' \rcolon N_2 \red N_3 \colon B }{
        \mbox{\begin{minipage}[t]{0.8\linewidth} \raggedright $ \Gam
            \vdash ((P \vcomp{M_2} P'), (Q \vcomp{N_2} Q')) \req ((P,
            Q) \vcomp{(M_2, N_2)} (P', Q')) $
            \\
            \raggedleft ${} \rcolon (M_1, N_1) \red (M_3, N_3) \colon
            A \times B$
          \end{minipage}} } \and %
      \inferrule{
        \Gam \vdash P \rcolon M_1 \red M_2 \colon A \times B \\
        \Gam \vdash Q \rcolon M_2 \red M_3 \colon A \times B }{ \Gam
        \vdash (\proj_{A, B} (P \vcomp{M_2} Q)) \req (\proj_{A, B} P
        \vcomp{\proj_{A, B} M_2} \proj_{A, B} Q) \rcolon M_1 \red M_3
        \colon A } \and %
      \inferrule{
        \Gam \vdash P \rcolon M_1 \red M_2 \colon A \times B \\
        \Gam \vdash Q \rcolon M_2 \red M_3 \colon A \times B }{ \Gam
        \vdash (\projj_{A, B} (P \vcomp{M_2} Q)) \req (\projj_{A, B} P
        \vcomp{\projj_{A, B} M_2} \projj_{A, B} Q) \rcolon M_1 \red
        M_3 \colon B } 
    }
  \end{framed}
  \caption{Equations on reductions (Lifting)}
  \label{fig:req2}
\end{figure}

\end{document}